\documentclass{llncs}
\usepackage{graphicx}
\usepackage{amsmath}
\usepackage{amssymb}
\usepackage{enumerate}
\usepackage{varioref}
\usepackage{charter,eulervm}

\title{\sc {Algorithms for the strong 
chromatic index of Halin graphs, distance-hereditary graphs and 
maximal outerplanar graphs}}

\author{
 Ton~Kloks\inst{1}
\and 
 Sheung-Hung~Poon\inst{1}
\and 
 Chin-Ting~Ung\inst{1}
\and 
 Yue-Li~Wang\inst{2}
}
\institute{
 Department of Computer Science\\
 National Tsing Hua University,
 No.~101, Sec.~2, Kuang Fu Rd., Hsinchu, Taiwan\\
 {\tt spoon@cs.nthu.edu.tw, wonderboy0915@gmail.com} 
\and 
 Department of Information Management\\
 National Taiwan University of Science and Technology\\
 No.~43, Sec.~4, Keelung Rd., Taipei, 106, Taiwan\\
 {\tt ylwang@cs.ntust.edu.tw}
}

\begin{document}

\maketitle

\begin{abstract}
We show that there exist linear-time algorithms that compute the 
strong chromatic index of Halin graphs, of 
maximal outerplanar graphs and of 
distance-hereditary  
graphs. 
\end{abstract}

\section{Introduction}

\begin{definition}
Let $G=(V,E)$ be a graph. A {\em strong edge coloring\/} of $G$ is a 
proper edge coloring such that no edge is adjacent to two edges of the 
same color. 
\end{definition}

Equivalently, a strong edge coloring of $G$ is a vertex coloring 
of $L(G)^2$, the square of the linegraph of $G$. 
The strong chromatic index of $G$ is the minimal integer $k$ such 
that $G$ has a strong edge coloring with $k$ colors. We denote the 
strong chromatic index of $G$ by $s\chi^{\prime}(G)$. 

Recently it was shown that the strong chromatic index 
is bounded by 
\[(2-\epsilon)\Delta^2\] for some $\epsilon >0$, where 
$\Delta$ is the maximal degree of the graph~\cite{kn:molloy}.%
\footnote{In their paper Molloy and Reed state that $\epsilon \geq 0.002$ 
when $\Delta$ is sufficiently large.}  
Earlier, Andersen showed that the strong chromatic index of a cubic graph 
is at most ten~\cite{kn:andersen}. 

Let $\mathcal{G}$ be the class of chordal graphs, or the class of 
cocomparability graphs, or the class of 
weakly chordal graphs. If $G \in \mathcal{G}$ then also 
$L(G)^2 \in \mathcal{G}$ and it follows that the strong chromatic index 
can be computed in polynomial time for these 
classes~\cite{kn:cameron,kn:cameron2,kn:cameron3}.  
Also for graphs of bounded treewidth 
there exists a polynomial time algorithm 
that computes the strong chromatic 
index~\cite{kn:salavatipour}.%
\footnote{This algorithm checks in $O(n(s+1)^t)$ time whether a 
partial $k$-tree has a strong edge coloring that uses at most $s$ colors. 
Here, the exponent $t=2^{4(k+1)+1}$.} 
 
\begin{definition}
Let $T$ be a tree without vertices of degree two. Consider a plane embedding 
of $T$ and connect the leaves of $T$ by a cycle that crosses no edges of $T$. 
A graph that is constructed in this way is called a {\em Halin graph\/}. 
\end{definition}

Halin graphs have treewidth at most three. Furthermore, if $G$ 
is a Halin graph of bounded degree, then also $L(G)^2$ has 
bounded treewidth 
and thus the strong chromatic index of $G$ can be computed in 
linear time. Recently, Ko-Wei Lih, {\em et al.\/}, 
proved that a cubic Halin graph 
other than one of the two `necklaces' $Ne_2$ (the complement of 
$C_6$) and $Ne_4$, has strong 
chromatic index at most 7. The two exceptions have 
strong chromatic index 9 and 8, respectively. If $T$ is the underlying 
tree of the Halin graph, and if $G \neq Ne_2$ and $G$ is 
not a wheel $W_n$ 
with $n \neq 0 \bmod{3}$, then Ping-Ying Tsai, {\em et al.\/}, 
show that the strong chromatic index is bounded by 
$s\chi^{\prime}(T)+3$. 
(See~\cite{kn:shiu2,kn:shiu} for earlier results that  
appeared in regular papers.%
\footnote{The results of Ko-Wei Lih and Ping-Ying Tsai, {\em et al.\/}, were 
presented at the Sixth Cross-Strait Conference on Graph Theory and 
Combinatorics which was held at the National Chiao Tung University 
in Taiwan in 2011.})  
 
\medskip 

If $G$ is a Halin graph then $L(G)^2$ 
has bounded rankwidth. 
In~\cite{kn:ganian} it is shown that there exists a 
polynomial algorithm that computes the chromatic number 
of graphs with bounded rankwidth, thus the strong chromatic 
index of Halin graphs can be computed in polynomial time. 
In passing, let us mention the following result. 
A class of graphs $\mathcal{G}$ is 
$\chi$-bounded if there exists a function $f$ such that 
$\chi(G) \leq f(\omega(G))$ for $G \in \mathcal{G}$. Here 
$\chi(G)$ is the chromatic number of $G$ and $\omega(G)$ is the 
clique number of $G$.  
Recently, Dvo\v{r}\'ak and Kr\'al showed that for every $k$, 
the class of graphs with rankwidth at most $k$ 
is $\chi$-bounded~\cite{kn:dvorak}. 
Obviously, the graphs $L(G)^2$ have a uniform $\chi$-bound for 
graphs $G$ in the class of Halin graphs. 

In Section~\ref{section halin}
we show that there exists a linear-time algorithm 
that computes the strong chromatic index of Halin graphs. 
In Section~\ref{section dh} we show that there exists a 
linear-time algorithm that computes the strong chromatic index 
of distance-hereditary graphs. In Section~\ref{section out} we 
show that there exists a linear-time algorithm that computes the 
strong chromatic index of maximal outerplanar graphs.  
  
\section{The strong chromatic index of Halin graphs}
\label{section halin}

The following lemma is easy to check. 

\begin{lemma}[Ping-Ying Tsai]
\label{basics}
Let $C_n$ be the cycle with $n$ vertices and let $W_n$ be the wheel 
with $n$ vertices in the cycle. 
Then 
\[s\chi^{\prime}(C_n) = \begin{cases} 
3 & \quad \text{if $n= 0 \bmod{3}$} \\
5 & \quad \text{if $n=5$} \\
4 & \quad \text{otherwise} 
\end{cases}
\quad  
s\chi^{\prime}(W_n)= \begin{cases}
n+3 & \quad \text{if $n=0 \bmod{3}$}\\
n+5 & \quad \text{if $n=5$}\\
n+4 & \quad \text{otherwise.}
\end{cases}\]
\end{lemma}

A double wheel is a Halin graph in which the tree 
$T$ has exactly two vertices 
that are not leaves. 

\begin{lemma}[Ping-Ying Tsai]
\label{basic2}
Let $W$ be a double wheel where $x$ and $y$ are 
the vertices of $T$ that are not leaves. 
Then $s\chi^{\prime}(T)=d(x)+d(y)-1$ where $d(x)$ and 
$d(y)$ are the degrees of $x$ and $y$. Furthermore, 
\[s\chi^{\prime}(W)= 
\begin{cases}
s\chi^{\prime}(T)+4=9 & \quad \text{if $d(x)=d(y)=3$, {\em i.e.\/}, if 
$W=\Bar{C_6}$} \\
s\chi^{\prime}(T)+2=d(y)+4 & \quad \text{if $d(y) > d(x)=3$} \\
s\chi^{\prime}(T)+1=d(x)+d(y) & \quad \text{if $d(y) \geq d(x) >3$.}
\end{cases}\]
\end{lemma}
  
Let $G$ be a Halin graph with tree $T$ and cycle $C$. 
Then obviously,  
\begin{equation}
\label{eq1}
s\chi^{\prime}(G) \leq s\chi^{\prime}(T) + s\chi^{\prime}(C).
\end{equation}
The linegraph of a tree is a claw-free blockgraph. Since a sun 
$S_r$ with $r > 3$ has a claw,  
$L(T)$ has no induced sun $S_r$ with $r > 3$. It follows 
that $L(T)^2$ is a chordal graph~\cite{kn:laskar} (see 
also~\cite{kn:cameron}; in this paper Cameron proves 
that $L(G)^2$ is chordal for any chordal graph $G$). Notice that 
\begin{equation}
\label{bound}
s\chi^{\prime}(T)=\chi(L(T)^2)=\omega(L(T)^2) \leq  2\Delta(G)-1 
\quad\Rightarrow\quad 
s\chi^{\prime}(G) \leq 2\Delta(G)+4.
\end{equation}

\subsection{Cubic Halin graphs}

In this subsection we outline a simple linear-time algorithm for 
the cubic Halin graphs. 

\begin{theorem}
\label{cubic case}
There exists a linear-time algorithm that computes the 
strong chromatic index of cubic Halin graphs.
\end{theorem}
\begin{proof}
Let $G$ be a cubic Halin graph with plane tree $T$ and cycle $C$. 
Let $k$ be a natural number. We describe a linear-time 
algorithm that 
checks if $G$ has a strong edge coloring with at most $k$ colors. 
By Equation (\ref{bound}) 
we may assume that $k$ is at most 10. Thus the 
correctness of this algorithm proves the theorem. 
 
\medskip
Root the tree $T$ at an arbitrary leaf $r$ of $T$.  
Consider a vertex $x$ in $T$. There is a unique path $P$ 
in $T$ from $r$ to $x$ in $T$. Define the subtree $T_x$ at $x$  
as the maximal connected subtree of $T$ that does not contain 
an edge of $P$. If $x=r$ then $T_x=T$.   

Let $H(x)$ be the subgraph of $G$ induced by the vertices of 
$T_x$. 
Notice that, 
if $x \neq r$ then the edges of $H(x)$ that are 
not in $T$ form a path $Q(x)$ of edges in $C$.   

\medskip
For $x \neq r$ define the boundary $B(x)$ 
of $H(x)$ as the following set 
of edges. 
\begin{enumerate}[\rm (a)]
\item 
\label{i1}
The unique edge of $P$ that 
is incident with $x$. 
\item 
\label{i2}
The two edges of $C$ that connect 
the path $Q(x)$ of $C$ with the rest of $C$.   
\item Consider the endpoints of the 
edges mentioned in (\ref{i1}) and (\ref{i2}) 
that are in $T_x$. 
Add the remaining two edges that are incident with each of 
these endpoints to $B(x)$. 
\end{enumerate}
Thus the boundary $B(x)$ consists of at most 9 edges. 
The following claim is easy to check. It proves 
the correctness of the algorithm described below. 
Let $e$ be an edge of $H(x)$. Let $f$ be an edge 
of $G$ that is not an edge of $H(x)$. If $e$ and $f$ are 
at distance at most 1 in $G$ then $e$ or $f$ is in 
$B(x)$.\footnote{Two edges in $G$ are at distance at most one 
if the subgraph induced by their endpoints is either $P_3$, or 
$K_3$ or $P_4$. We assume that it can be checked in constant time 
if two edges $e$ and $f$ are at distance at most one. This can be 
achieved by a suitable data structure.}
  
\medskip
Consider all possible colorings of the edges in $B(x)$. 
Since $B(x)$ contains at most 9 edges and since there are at 
most $k$ different colors for each edge, there are at most 
\[k^9 \leq 10^9\]
different colorings of the edges in $B(x)$.       

The algorithm now fills a table which gives a boolean 
value for each coloring of the boundary $B(x)$. This boolean value is 
{\tt TRUE} if and only if  
the coloring of the edges in $B(x)$ extends to an  
edge coloring of the union of the sets of edges in $B(x)$ 
and in $H(x)$ with at most $k$ colors, such that any pair of 
edges in this set that are 
at distance at most one in $G$, have different colors. 
These boolean values are computed as follows. We prove the correctness 
by induction on the size of the subtree at $x$. 

\medskip
First consider the case where the subtree at $x$ consists of the 
single vertex $x$. Then $x \neq r$ and $x$ is a leaf of $T$. 
In this case  
$B(x)$ consists 
of three edges, namely the three edges that are incident with $x$. 
These are two edges of $C$ and one edge of $T$. 
If the colors of these three edges in 
$B$  
are different 
then the boolean value is set to 
{\tt TRUE}. Otherwise it is set to {\tt FALSE}. 
Obviously, this is a correct assignment. 

\medskip
Next consider the case where $x$ is an internal vertex of $T$. 
Then $x$ has two children in the subtree at $x$. 
Let $y$ and $z$ be the two children and consider the two 
subtrees rooted at $y$ and $z$. 

The algorithm that computes the tables for each vertex $x$ 
processes the subtrees in order of increasing number of vertices.   
(Thus the roots of the subtrees are visited in postorder). 
We now assume that the tables at $y$ and $z$ are 
computed correctly and show how the  
table for $x$ is computed correctly and in constant time. 
That is, we prove that the algorithm described below 
computes the table at $x$ such that it contains a 
coloring of $B(x)$ with a value {\tt TRUE} if and only 
if there 
exists an extension of this coloring to the 
edges of $H(x)$ and $B(x)$ such that any two 
different edges $e$ and $f$ 
at distance at most one in $G$, each one in 
$H(x)$ or in $B(x)$, have different colors. 
   
\medskip
Consider a coloring of the edges in the 
boundary $B(x)$. 
The boolean value in the table of $x$ 
for this coloring is computed as follows. 
Notice that 
\begin{enumerate}[\rm (i)]
\item $B(y) \cap B(z)$ consists of one edge and this 
edge is not in $B(x)$, and 
\item $B(x) \cap B(y)$ consists 
of at most four edges, namely the edge $(x,y)$ and the three 
edges of $B(y)$ that are incident with one vertex of 
$C \cap H(y)$.   
Likewise, $B(x) \cap B(z)$ consists of at most four edges. 
\end{enumerate} 
The algorithm varies the possible colorings of 
the edge in $B(y) \cap B(z)$. 
Colorings of $B(x)$, $B(y)$ and $B(z)$ are 
consistent if the intersections are the same color and the pairs 
of edges in 
\[B(x) \cup B(y) \cup B(z)\] 
that are at distance at most one in $G$ have different colors.   
A coloring of $B(x)$ is assigned the value {\tt TRUE} 
if there exist colorings of $B(y)$ and $B(z)$ such that the 
three colorings are consistent and $B(y)$ and $B(z)$ are assigned the 
value {\tt TRUE} in the tables at $y$ and at $z$ respectively. 
Notice that the table at $x$ is built in constant time. 

Consider a coloring of $B(x)$ that is assigned the value {\tt TRUE}. 
Consider colorings of the edges of $B(y)$ and $B(z)$ that are 
consistent with $B(x)$ and that are assigned the value 
{\tt TRUE} in the tables at $y$ and $z$. By induction, there 
exist extensions of the colorings of $B(y)$ and $B(z)$ 
to the edges of $H(y)$ and $H(z)$. The union of these 
extensions provides a $k$-coloring of the edges in $H(x)$. 

Consider two edges $e$ and $f$ in $B(x) \cup B(y) \cup B(z)$. 
If their distance 
is at most one then they have different colors since 
the coloring of $B(x) \cup B(y) \cup B(z)$ is 
consistent. Let $e$ and $f$ be a pair of edges in $H(x)$. 
If they are both in $H(y)$ or both in $H(z)$ then they have 
different colors.  
Assume that $e$ is in $H(y)$ and assume that $f$ is not in 
$H(y)$. If $e$ and $f$ are at distance at most one, 
then $e$ or $f$ is in $B(y)$. If they are both in $B(y)$, 
then they have different colors, due to the consistency. 
Otherwise, by the induction hypothesis, they have different colors. 
This proves the claim on the correctness. 
    
\medskip
Finally, consider the table for the vertex $x$ which is 
the unique neighbor of $r$ in $T$. 
By the induction hypothesis, 
and the fact that every edge in $G$ is either in 
$B(x)$ or in $H(x)$,   
$G$ has a strong edge coloring with at most 
$k$ colors if and only if the table at $x$ contains a coloring 
of $B(x)$ with three different colors for which the boolean is 
set to {\tt TRUE}. 

This proves the theorem. 
\qed\end{proof}

\begin{remark}
The involved constants in this algorithm are improved considerably 
by the recent results of Ko-Wei Lih, Ping-Ying Tsai, {\em et al.\/}. 
\end{remark}
  
\subsection{Halin graphs of general degree}

\begin{theorem}
\label{general Halin graphs}
There exists a linear-time algorithm that computes the strong 
chromatic index of Halin graphs.
\end{theorem}
\begin{proof}
The algorithm is similar to the algorithm for 
the cubic case.   

\medskip
Let $G$ be a Halin graph, let $T$ be the underlying plane tree,  
and let $C$ be the cycle that connects the leaves of $T$. 
Since $L(T)^2$ is chordal the chromatic number of $L(T)^2$ 
is equal to the 
clique number of $L(T)^2$, which is 
\[s\chi^{\prime}(T)= \max \;\{\;d(u)+d(v)-1\;|\; (u,v) \in E(T)\;\},\]
where $d(u)$ is the degree of $u$ in the tree $T$. 
By Formula~(\ref{eq1}) and Lemma~\ref{basics} 
the strong chromatic index of $G$ is one of 
the six possible values%
\footnote{Actually, according to the recent results of 
Ping-Ying Tsai, {\em et al.\/}, 
the strong chromatic index of $G$ is at most $s\chi^{\prime}(T)+3$ 
except when $G$ is a wheel or $\Bar{C_6}$.}
\[s\chi^{\prime}(T), s\chi^{\prime}(T)+1, \ldots, s\chi^{\prime}(T)+5.\]

\medskip
Root the tree at some leaf $r$ and consider a subtree $T_x$ at a 
node $x$ of $T$. 
Let $H(x)$ be the subgraph of $G$ induced by the vertices of $T_x$.  
Let $y$ and $z$ be the two boundary vertices of $H(x)$ in $C$. 

\medskip
We distinguish the following six types of edges corresponding 
to $H(x)$. 
\begin{enumerate}[\rm 1.]
\item The set of edges in $T_x$ that are adjacent to $x$. 
\item The edge that connects $x$ to its parent in $T$. 
\item The edge that connects $y$ to its neighbor in $C$ that is 
not in $T_x$. 
\item The set of edges in $H(x)$ that have endpoint $y$. 
\item The edge that connects $z$ to its neighbor in $C$ that is 
not in $T_x$. 
\item The set of edges in $H(x)$ that have endpoint $z$. 
\end{enumerate}
When $x$ is adjacent to $y$ then we make a separate type 
for the edge $(x,y)$ and similar in the case where $x$ is 
adjacent to $z$. 

Notice that the set of edges of every type has bounded 
cardinality, except the first type. 

\medskip
Consider a $0/1$-matrix $M$ with rows indexed by the six 
to eight  
types of edges and columns indexed by the colors. 
A matrix entry $M_{ij}$ is 1 if there is an edge of the row-type 
$i$ that is colored with the color $j$ and otherwise this 
entry is 0. Since $M$ has at most 8 rows, the rank over 
$GF[2]$ of $M$ is at most 8.  

\medskip 
Two colorings are equivalent if there is a 
permutation of the colors that maps one coloring to the other one. 
Let $S \subseteq \{1,\ldots,8\}$ and let $W(S)$ be the 
set of colors that are used by 
edges of type $i$ for all $i \in S$.  
A class of equivalent colorings is fixed by the set 
of cardinalities  
\[\{\; |W(S)| \;|\; S \subseteq \{1,\ldots,8\} \;\}.\]

\medskip
We claim that the number of equivalence classes is 
constant. The number of ones in the row of the first type 
is the degree of $x$ in $H(x)$. Every other row has at most 
3 ones. This proves the claim. 

\medskip 
Consider the union of two subtrees, say at $x$ and $x^{\prime}$. 
The algorithm considers all equivalence classes of colorings of the 
union, and checks, by table look-up, whether it decomposes into 
valid colorings of $H(x)$ and $H(x^{\prime})$. An easy way to do this 
is as follows. First double the number of types, by distinguishing 
the edges of $H(x)$ and $H(x^{\prime})$. Then enumerate 
all equivalence classes of colorings. Each equivalence class 
is fixed by a sequence of $2^{16}$ numbers, as above. 
By table look-up,  
check if an equivalence class restricts to a valid coloring 
for each of $H(x)$ and $H(x^{\prime})$.  
Since this takes constant time, the algorithm runs in linear time. 

This proves the theorem. 
\qed\end{proof}

\section{Distance-hereditary graphs}
\label{section dh}

\begin{definition}[\cite{kn:howorka}]
A graph $G$ is {\em distance hereditary\/} if any two 
nonadjacent vertices in a component of any induced subgraph $H$ 
are at the same distance in $H$ as they are in the graph $G$. 
\end{definition}

In other words, any two chordless paths between two 
nonadjacent vertices is of the same length. 
Distance-hereditary graphs are exactly the graphs that have 
rankwidth one~\cite{kn:chang}. 
In this section we prove that there is a linear-time algorithm 
that computes the strong chromatic index of distance-hereditary 
graphs. 
Distance-hereditary graphs are perfect. They are the graphs without 
induced gem, house, hole or domino. Cameron proves in~\cite{kn:cameron3} 
that, for $k \geq 4$, if $G$ has no induced cycles of length more 
than four then also $L(G)^2$ has no such induced cycles. It follows 
that, if $G$ is distance hereditary then $L(G)^2$ is perfect. 
Therefore, to compute the chromatic number of 
$L(G)^2$ it suffices to compute the clique number. 

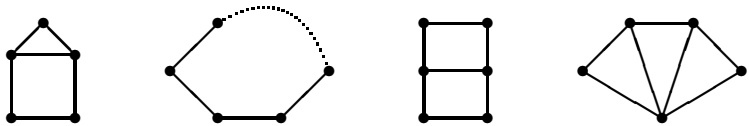
\begin{figure}[htbp]
\setlength{\unitlength}{1.2pt}
\begin{center}
\begin{picture}(230,35)
\thicklines
\put(0,0){\circle*{3.0}} \put(20,0){\circle*{3.0}}
\put(0,20){\circle*{3.0}} \put(20,20){\circle*{3.0}}
\put(10,30){\circle*{3.0}} \put(0,0){\line(1,0){20}}
\put(0,0){\line(0,1){20}} \put(20,20){\line(-1,0){20}}
\put(20,20){\line(0,-1){20}} \put(10,30){\line(-1,-1){10}}
\put(10,30){\line(1,-1){10}}
\put(50,15){\circle*{3.0}} \put(65,0){\circle*{3.0}}
\put(65,30){\circle*{3.0}} \put(85,0){\circle*{3.0}}
\put(100,15){\circle*{3.0}} \put(50,15){\line(1,-1){15}}
\put(50,15){\line(1,1){15}} \put(65,0){\line(1,0){20}}
\put(85,0){\line(1,1){15}} \qbezier[25](65,30)(90,45)(100,15)
\put(130,0){\circle*{3.0}}\put(150,0){\circle*{3.0}}
\put(130,15){\circle*{3.0}}\put(150,15){\circle*{3.0}}
\put(130,30){\circle*{3.0}}\put(150,30){\circle*{3.0}}
\put(130,0){\line(0,1){30}} \put(150,0){\line(0,1){30}}
\put(130,0){\line(1,0){20}}\put(130,15){\line(1,0){20}}
\put(130,30){\line(1,0){20}}
\put(180,15){\circle*{3.0}} \put(195,30){\circle*{3.0}}
\put(205,0){\circle*{3.0}} \put(215,30){\circle*{3.0}}
\put(230,15){\circle*{3.0}} \put(205,0){\line(-5,3){25}}
\put(205,0){\line(-1,3){10}}\put(205,0){\line(1,3){10}}
\put(205,0){\line(5,3){25}} \put(180,15){\line(1,1){15}}
\put(195,30){\line(1,0){20}} \put(215,30){\line(1,-1){15}}
\end{picture}
\caption{A graph is distance hereditary if it has no induced
house hole, domino or gem.}
\label{HHDG}
\end{center}
\end{figure}

A pendant vertex in a graph is a vertex of degree one. 
A twin is a pair of vertices $x$ and $y$ with 
the same open or the same closed neighborhood. 
When $x$ and $y$ are adjacent then the twin is called 
a true twin and otherwise it is called a false twin. 
A $P_4$ is a path with four vertices. 

\begin{theorem}[\cite{kn:bandelt}]
A graph $G$ is distance hereditary if and only if $G$ 
is obtained from an edge by a sequence of the following 
operations. 
\begin{enumerate}[\rm (a)]
\item Creation a pendant vertex. 
\item Creation of a twin. 
\end{enumerate}
\end{theorem}

\begin{lemma}
\label{false twin}
Let $G$ be a graph and consider the graph $G^{\prime}$ 
obtained from $G$ by creating a false twin $x^{\prime}$ 
of a vertex $x$ in $G$. 
Then $L(G^{\prime})^2$ is obtained from $L(G)^2$ by a series 
of true twin operations. 
\end{lemma}
\begin{proof}
Let $a_1,\ldots,a_s$ be the neighbors of $x$ in $G$. 
By definition of 
$L(G)^2$, each edge $(x^{\prime},a_i)$ is a true twin of 
the edge $(x,a_i)$ in $L(G^{\prime})^2$.  
\qed\end{proof}

\begin{definition}
A graph $G$ is a {\em cograph\/} if $G$ has no 
induced $P_4$. 
\end{definition}

A cograph is obtained from a graph consisting of one vertex 
by a series of twin operations. 
Chordal cographs are the graphs without induced $P_4$ and $C_4$. 
These are also called trivially perfect. 

\begin{lemma}
\label{cograph}
If $G$ is a cograph then $L(G)^2$ is trivially perfect.  
\end{lemma}
\begin{proof}
A cograph with at least two vertices is either the 
join or the union of two cographs $G_1$ and $G_2$. 
Assume that $G$ is the join of two cographs $G_1$ and $G_2$. 
The set of edges with one endpoint in $G_1$ and the other in 
$G_2$ are a clique in $L(G)^2$. Furthermore, this set of 
edges is adjacent to every edge that is 
contained in $G_i$ for $i \in \{1,2\}$. 
In other words, every component of $L(G)^2$ has a universal vertex, 
{\em i.e.\/}, a vertex adjacent to all other vertices. 
The graphs that satisfy this property are exactly the graphs 
in which every component is the comparability graph of a tree and 
these are exactly the graphs without 
induced $P_4$ and $C_4$~\cite{kn:wolk}.  
\qed\end{proof}
 
Notice that Lemma~\ref{cograph} provides a linear-time 
algorithm for computing the strong chromatic index 
of cographs. A cotree decomposition can be obtained in 
linear time. Assume that $G$ is the join of two cographs 
$G_1$ and $G_2$. Then every edge with both ends in $G_1$ 
is adjacent in $L(G)^2$ 
to every edge with both ends in $G_2$. 
Let $X$ be the set of edges 
with one endpoint in $G_1$ and the other endpoint in 
$G_2$. By dynamic programming on the cotree, 
compute the clique numbers 
of $L(G_1)^2$ and $L(G_2)^2$. Add $|X|$ to the sum of both. 
If $G$ is the union of $G_1$ and $G_2$ then the 
strong chromatic index of $G$ is the maximum of the 
clique numbers of 
$L(G_1)^2$ and $L(G_2)^2$. 
This proves the following theorem. 

\begin{theorem}
There exists a linear-time algorithm that computes the 
strong chromatic index of cographs. 
\end{theorem}

\begin{lemma}
\label{DH}
If $G$ is distance hereditary then every 
neighborhood in $L(G)^2$ induces a trivially perfect graph. 
\end{lemma}
\begin{proof}
We prove the theorem by induction on the elimination 
ordering of $G$ by pendant vertices and elements of twins. 

First, assume that $G^{\prime}$ is 
obtained from $G$ by creating a false twin 
$x^{\prime}$ of a vertex $x$ in $G$. 
By 
Lemma~\ref{false twin} $L(G^{\prime})^2$ is obtained 
from $L(G)^2$ by a series of true twin operations. 
In that case the claim follows easily, by induction.  

Secondly, 
consider the operation which adds a pendant vertex $x^{\prime}$, 
made adjacent to a vertex $x$ in $G$. Let $a_1,\ldots,a_s$ 
be the neighbors of $x$ in $G$.  
Notice that the adjacencies of the edge $(x,x^{\prime})$ in 
$L(G^{\prime})^2$ are of 
the following types of edges in $G$. 
\begin{enumerate}[\rm (a)]
\item 
\label{a}
All edges $(x,a_i)$, $i \in \{1,\ldots,s\}$. 
Call this set of edges $X$. 
\item
\label{b}
 The edges $(a_i,a_j) \in E(G)$, for $i,j \in \{1,\ldots,s\}$. 
\item
\label{c} 
Edges $(a_i,u)$, for $i \in \{1,\ldots,s\}$ and 
$u \in N_G(a_i) \setminus N_G[x]$. 
\end{enumerate}
Call two vertices in $N_G(x)$ equivalent if they have 
the same neighbors in the graph $G-N_G[x]$. Since there is no 
house, hole, domino or gem every equivalence class 
is joined to or disjoint from every other equivalence class. 
Let $H$ be the graph with vertex set the set of 
equivalence classes and edge set the 
pairs of equivalence classes that are joined. 
Since $G$ has no gem, the graph $H$ 
has no induced $P_4$ and so it is a cograph. 
Furthermore, by Lemma~\ref{cograph} $L(H)^2$ is 
trivially perfect. 
 
Consider the components of $G-N_G[x]$. For 
any two components 
$C_1$ and $C_2$ their neighborhoods $N_G(C_1)$ and 
$N_G(C_2)$ are either disjoint or ordered by 
inclusion. 
First consider the components that have a 
maximal neighborhood in $N_G(x)$ and remove all 
other components. Consider the equivalence 
classes defined by these components. The graph on these 
equivalence classes is a cograph and the square of the 
linegraph is a chordal cograph.      
Next, consider such an equivalence class $Q$ 
with at least two vertices. Remove the components $C$ of 
$G-N_G[x]$ with $N(C)=Q$. If there are some 
components left of which the neighborhood is 
properly contained in $Q$ then partition the vertices 
of $Q$ into secondary equivalence classes. If there are 
no more components with their neighborhood 
contained in $Q$ then define the secondary equivalence 
classes as sets of single vertices. As above, for each 
equivalence class $Q$,  
the secondary equivalence classes form a cograph $H_Q$. 
Also, $L(H_Q)^2$ is trivially perfect. Continuation of this  
process defines a chordal cotree on the subgraph of 
$L(G)^2$ induced by the edges of 
types~(\ref{b}) and~(\ref{c}). 
Notice that the set $X$ of edges is universal 
in the neighborhood of $(x,x^{\prime})$ in $L(G)^2$.    

Finally, consider the case where $G^{\prime}$ is 
obtained from $G$ by creating a true twin $x^{\prime}$ 
of a vertex $x$ in $G$. Subdivide this operation into 
two steps. First create a false twin. Let $G^{\ast}$ 
be the graph obtained in this manner. We proved above 
that every neighborhood in $L(G^{\ast})^2$ is 
trivially perfect. Secondly, adding the edge $(x,x^{\prime})$ 
to $G^{\ast}$ is similar to the operation of adding a 
pendant vertex. The set $X$ of edges 
as described above, now consists of pairs of true twins 
in $L(G)^2$. 
The other types of adjacencies of $(x,x^{\prime})$,  
as described in~(\ref{b}) and~(\ref{c}),  
are the same as above.   

This proves the lemma.  
\qed\end{proof}

\begin{theorem}
There exists a linear-time algorithm that computes 
the strong chromatic index of distance-hereditary graphs. 
\end{theorem}
\begin{proof}
Let $G$ be distance hereditary. Consider a rank decomposition 
of $G$ of rankwidth one. 
This is a pair $(T,f)$ where 
$T$ is a rooted binary tree and where $f$ is a bijection 
from the vertices in $G$ to the leaves of $T$. Consider a subtree 
$T_e$ of $T$ rooted at some edge $e$ of $T$. Define $G_e$ as 
the subgraph  
of $G$ 
induced by the vertices that are mapped to leaves in $T_e$. 
Let $S_e$ be the set of vertices of $G_e$ that have 
neighbors in $G-V(G_e)$. The set $S_e$ is called the twinset 
of $G_e$~\cite{kn:chang}. 
All vertices of $S_e$ have the same neighbors in 
$G-V(G_e)$. 

Consider an edge $e$ of $T$ and let $e_1$ and $e_2$ be the two 
children of $e$ in $T$. The graph $G_e$ is obtained from 
$G_{e_1}$ and $G_{e_2}$  
by a join or by a union of the twinsets $S_{e_1}$ and 
$S_{e_2}$. The twinset $S_e$ of $G_e$ is either one of 
$S_{e_1}$ and $S_{e_2}$ or it is the union of the two~\cite{kn:chang}. 

Let $e$ be a line in $T$ with children $e_1$ and $e_2$. 
Let $S_1$ and $S_2$ be the twinsets of $G_{e_1}$ and of $G_{e_2}$ 
and assume that there is a join between $S_1$ and $S_2$. 
Let $X$ be the set of edges between $S_1$ and $S_2$.  
For $i \in \{1,2\}$ 
choose a maximal clique $\Omega_i$ in each $L(S_i)^2$ such that 
the set of end-vertices of edges in $\Omega_i$ has a maximal number 
of neighbors in $G_{e_i}-S_i$. Let $\omega_i$ 
be the number of edges in this maximal clique. 
Let $N_i$ be this number of neighbors. 
The algorithm keeps track of the maximal value of 
$|X|+\omega_1+\omega_2+N_1+N_2$.   
It is easy to see that this algorithm can be implemented 
to run in linear time. 
\qed\end{proof}

\section{Maximal outerplanar graphs}
\label{section out}

A maximal outerplanar graph $G$ is a ternary tree $T$ ({\em i.e.\/}, 
every vertex in $T$ has degree at most three) of triangles, 
where two triangles that are adjacent in the tree 
share an edge (see~{\em e.g.},~\cite{kn:kloks}).  

\medskip 

In~\cite{kn:hocquard} Hocquard, {\em et al.\/}, prove that 
for every outerplanar graph $G$ with maximal degree $\Delta \geq 3$, 
\[s\chi^{\prime}(G) \leq 3(\Delta-1).\]
They also prove, among various other 
NP-completeness results, that strong edge 4-coloring is NP-complete 
for planar bipartite graphs with maximal degree three and any 
fixed girth. 

\medskip 
 
\begin{definition}
An {\em extended triangle\/} in a maximal outerplanar graph 
consists of a triangle plus all the edges that are 
incident with some vertex of the triangle. 
\end{definition}
 
All edges of an extended triangle must be colored 
different. Let $\phi$ be the maximal number of 
edges of all extended triangles in $G$. 
In the following theorem we prove that 
there exists a strong edges coloring that uses $\phi$ colors. 

\begin{theorem}
There exists a linear-time algorithm that computes 
the strong chromatic index of maximal outerplanar graphs. 
\end{theorem}
\begin{proof}
The algorithm colors the edges in a greedy manner as follows.
First we make one leaf node of $T$ as the root.
Then we traverse $T$ in a breath-first manner.
When we reach a node $v$, we color the uncolored edges 
of its corresponding
extended triangle $\tau$ so that the colors 
used for uncolored edges are different
from the colors of those colored edges in $\tau$.
As the number of edges of $\tau$ is at most $\phi$,
$\phi$ colors are sufficient to do the 
coloring of edges in $\tau$.
We proceed to color the edges in other extended triangles 
for other nodes in $T$
via the breath-first search traversal order. 
At the end of the traversal,
we finish the coloring of all edges in $G$ using only $\phi$ colors.

\smallskip 

For the correctness of the algorithm, we argue as follows.
It is easy to see that for any two edges of $(e,e')$ within distance two,
they must both appear in some extended triangle $\tau$ 
which implies that they
will obtain different colors when $\tau$ is visited  
in the breadth-first
traversal. 
Hence, the above coloring is thus a strong edge coloring for $G$,
and we obtain that $s\chi^{\prime}(G) = \phi$.
\qed\end{proof}
      
\section{Concluding remarks}

If $G$ is a circular-arc graph then $L(G)^2$ is also 
a circular-arc graph~\cite{kn:golumbic3}. 
Unfortunately, coloring a circular-arc graph is 
NP-complete~\cite{kn:marx}.  
When $G$ is AT-free then also $L(G)^2$ is AT-free~\cite{kn:cameron2}. 
As far as we know the complexity of coloring AT-free graphs is 
an open problem. There is some hope, since the maximum independent set 
problem is polynomial for this class of graphs. 

For Halin graphs we tried to prove that there is an 
optimal strong edge-coloring such that the edges in the cycle 
can be colored with colors from a fixed set of constant size. 
If true, then this would probably improve the timebound for 
the strong chromatic index problem on Halin graphs. 
Moser and Sikdar prove that the maximum induced matching problem 
on planar graphs is fixed-parameter tractable~\cite{kn:moser}.  
As far as we know the parameterized complexity of 
the strong chromatic index problem on planar graphs is open. 
Computing a maximum induced matching in planar graphs, or 
in bipartite graphs is 
NP-complete~\cite{kn:cameron2}.  

\medskip 

An example of a distance-hereditary graph $G$ for 
which $L(G)^2$ is not chordal is depicted in 
Figure~\ref{counterexample}. 

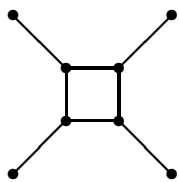
\begin{figure}
\setlength{\unitlength}{2pt}
\begin{center}
\begin{picture}(30,30) 
\thicklines
\put(0,0){\circle*{2.0}}
\put(30,0){\circle*{2.0}}
\put(10,10){\circle*{2.0}}
\put(20,10){\circle*{2.0}}
\put(10,20){\circle*{2.0}}
\put(20,20){\circle*{2.0}}
\put(0,30){\circle*{2.0}}
\put(30,30){\circle*{2.0}}
\put(0,0){\line(1,1){10}}
\put(10,10){\line(1,0){10}}
\put(10,10){\line(0,1){10}}
\put(20,10){\line(1,-1){10}}
\put(20,10){\line(0,1){10}}
\put(10,20){\line(1,0){10}}
\put(10,20){\line(-1,1){10}}
\put(20,20){\line(1,1){10}}
\end{picture}
\caption{A distance-hereditary graph $G$ for 
which $L(G)^2$ is not chordal.} 
\label{counterexample}
\end{center}
\end{figure}

Probably the following conjecture is true. 
If that is the case then the strong chromatic index 
can be computed in polynomial time for 
graphs of bounded rankwidth~\cite{kn:dvorak,kn:ganian}. 

\begin{conjecture}
There exists a function $\rho: \mathbb{N} \rightarrow \mathbb{N}$ 
for which the following holds. Let $G$ be a graph 
of rankwidth $k$. Then the rankwidth of $L(G)^2$ 
is at most $\rho(k)$. 
\end{conjecture}

\section{Acknowledgements}

Ton Kloks, Sheung-Hung Poon and Chin-Ting Ung 
thank the National Science Council of Taiwan for their support. 
Ton Kloks is supported under grant 
NSC~99--2218--E--007--016.
Sheung-Hung Poon and Chin-Ting Ung are supported under grants 
NSC 99--2218--E--007--016, NSC 100--2218--E--007--007 
and NSC 100--2628--E--007--020--MY3.

\end{document}